\documentclass[aps,twocolumn,prl,superscriptaddress,floatfix,letter]{revtex4}

\usepackage{amssymb,amsmath,graphicx,bbm,color}
\allowdisplaybreaks
\usepackage{verbatim}
\usepackage{relsize}
\usepackage{hyperref}

\newtheorem{lemma}{Lemma}[section]
\newenvironment{proof}{\textit{Proof. }}

\usepackage{dsfont}

\newcommand{\tr}{\operatorname{tr}}

\usepackage{dsfont}

\newcommand{\specialcell}[2][c]{ \begin{tabular}[#1]{@{}c@{}}#2\end{tabular}}

\newcommand{\bgt}{\begin{itemize}}
\newcommand{\ent}{\end{itemize}}

\newcommand{\op}{\operatorname}

\newcommand{\la}{\label}

\newcommand{\lan}{\langle}
\newcommand{\ran}{\rangle}

\newcommand{\ds}{\displaystyle}

\newcommand{\Tr}{\operatorname{Tr}}

\newcommand{\E}{\op{\mathbb{E}}}

\newcommand{\f}{\frac}

\newcommand{\bbm}{\begin{bmatrix}}
\newcommand{\ebm}{\end{bmatrix}}
\newcommand{\bes}{\begin{equation*}}
\newcommand{\ees}{\end{equation*}}
\newcommand{\be}{\begin{equation}}
\newcommand{\ee}{\end{equation}}
\newcommand{\beqy}{\begin{eqnarray}}
\newcommand{\eeqy}{\end{eqnarray}}
\newcommand{\beq}{\begin{eqnarray*}}
\newcommand{\eeq}{\end{eqnarray*}}
\newcommand{\one}{\mathbbm{1}}

\newcommand{\bpm}{\begin{pmatrix}}
\newcommand{\epm}{\end{pmatrix}}

\newcommand{\tblue}{\textcolor{blue}}

\long\def\symbolfootnote[#1]#2{\begingroup
\def\thefootnote{\fnsymbol{footnote}}\footnote[#1]{#2}\endgroup}

\renewcommand\Re{\operatorname{Re}}
\renewcommand\Im{\operatorname{Im}}

\begin{document}

\setkeys{Gin}{width=0.9\textwidth}

\title{Thermalisation of a quantum system from first principles}

\author{Gr\'egoire Ithier}      \affiliation{Department of Physics, Royal Holloway, University of London, United Kingdom}
\author{Florent Benaych-Georges}     \affiliation{MAP5, UMR CNRS 8145 -- Universit\'e Paris Descartes, France}



%



%
%
%
%


%
%
%
%

%
%
%

%
%



\maketitle

{\bf
Why is thermalisation a universal phenomenon? How does a quantum system reach thermodynamical equilibrium?
These questions are not new, dating even from the very birth of quantum theory~\cite{vonNeumann},
 and have been 
the subject of a renewed interest over the two last decades (see the review in~\cite{eisert_quantum_2015} and references therein).
Regarding first  the stationary limit of thermalisation, 
progress has been made demonstrating the \textit{typical} behavior of local
equilibrium properties of large closed quantum  
systems\cite{tasaki_quantum_1998,popescu_entanglement_2006,goldstein_canonical_2006,linden_quantum_2009,reimann_foundation_2008,riera_thermalization_2012},
and the connection of the hypothesis of eigenstate thermalisation (ETH~\cite{vonNeumann,deutsch_quantum_1991,srednicki_chaos_1994,rigol_thermalization_2008}) with this typicality property~\cite{dalessio_quantum_2015}.
%
 Regarding the dynamics  of thermalisation, which can be considered more generally in the unifying framework of measurement and decoherence~\cite{zurek_decoherence_2003},
 models
 specific to certain physical systems have been extensively studied
 using various calculation techniques and various assumptions on the initial state of the largest subpart (the "environment"), the most popular being an environment already in a thermal
 state~\cite{Weiss} (see~\cite{genway_dynamics_2013} for a \textit{non} thermal initial state).
In this Letter, we propose a universal model demonstrating that thermalisation of a small quantum system is an emergent property of the
 unitary evolution under a Schr\"{o}dinger equation of a larger composite system, whose initial state can be \textit{arbitrary}.
We show that the origin of universality lies in the phenomenon of ``measure concentration"~\cite{MR1387624,MR856576},
which provides self-averaging properties for the reduced density matrix characterizing the state of the small subsystem.
Using our framework, we focus on the asymptotic state at long times and consider its stationary properties.
 In typical macroscopic conditions, we recover the canonical state and the Boltzmann distribution 
 well known from statistical thermodynamics.
 This findings lead us to propose an alternative and more general definition of the canonical partition function which
also allow us to describe non thermal stationary states.
}




Thermalisation is probably the most common phenomenon one can encounter in nature. Its 
universality, i.e. the fact that it does not depend on microscopic details
 but only on a small set of parameters  (like temperature or pressure) defining a 
 \textit{macrostate}, 
and its irreversible character have been known experimentally for a long time.
In typical conditions, non equilibrium dynamics is expected to lead to some stationary state, independent of initial conditions
where macroscopic quantities can be calculated using statistical thermodynamics.
However, despite some very recent progress
\cite{popescu_entanglement_2006,goldstein_canonical_2006,linden_quantum_2009,reimann_foundation_2008} the foundations of
 this statistical framework
are still relying on a set of assumptions 
 where the role of randomness and the associated lack of knowledge, the role of averaging over this randomness and
the supposed link with temporal averages through ergodicity, is not justified in a satisfactory manner (see for instance
 the introductory discussion in~\cite{gemmer_distribution_2003} and references therein). 
With the recent progress in the quantum engineering of mesoscopic systems like ultra cold atomic 
gases~\cite{bloch_quantum_2012} or superconducting circuits~\cite{DevoretSchoelkopf}
 allowing
 %
  the simulation of (almost) closed quantum systems, 
these questions have undergone a renewed interest: understanding how a \textit{local} thermodynamical equilibrium 
can emerge from the unitary quantum evolution of a large composite system prepared initially in a pure state
  is becoming more urgent.
In this Letter, we propose and study a new model capturing both the transient and the stationary regimes of the time evolution of an
  arbitrary quantum system in contact with an arbitrary quantum environment. This composite system can be initially prepared in
   an arbitrary state: there is no thermal 
  equilibrium hypothesis for the environment in our model, since our present goal is to set firm ground for quantum
   statistical thermodynamics from first principles.
   Focusing on the state of the small subsystem in the long time limit, we are able to propose
 an alternative and more general definition of the canonical partition function describing this  state. 
We calculate this generalized partition function for a certain class of interaction Hamiltonians, and find a simple interpretation 
of thermodynamical equilibrium involving the ratios between densities of states of the environment and of the composite system.
Then we consider the case 
where a temperature can be defined and recover the well known
canonical distribution.
Our method advances the use of two key ingredients: the phenomenon of ``measure concentration'',
which provides \emph{self-averaging} properties of reduced density matrices, and the calculation of the statistics of the \textit{overlaps} 
between eigenvectors of a ``bare" and a ``dressed" Hamiltonian. 



We consider a system $S$ in contact with an environment $E$  and  denote by $\mathcal{H}_s$, $\mathcal{H}_e$  their respective Hilbert spaces. 
The composite system $S+E$ is a closed system whose Hilbert space is the tensor product $\mathcal{H}=\mathcal{H}_s \otimes \mathcal{H}_e$ and whose total Hamiltonian $\hat{H}$ is $\hat{H}=\hat{H}_s+\hat{H}_e+\hat{W}$ where $\hat{W}$ is an interaction term.
Eigenvectors of the ``bare'' Hamiltonian $\hat{H}_s+\hat{H}_e$ are written as $|\phi_n\rangle$ and are tensor products of eigenvectors $|\epsilon_s\rangle$ of $\hat{H}_s$ and
eigenvectors $| \epsilon_e\rangle$ of $\hat{H}_e$, with the eigenenergy $\epsilon_n=\epsilon_s+\epsilon_e$.
We write $|\psi_i\rangle$ for the eigenvectors of the total or ``dressed'' Hamiltonian $\hat{H}_s+\hat{H}_e+\hat{W}$ and $\{ \lambda_i \}_i$ the set of associated eigenvalues. The state of $S+E$ is described by a density matrix $\varrho(t)$ obeying the following relation derived from the Schr\"odinger equation: 
$$ \varrho(t) =\hat{U}_t  \varrho(0) \hat{U}_t^\dagger \quad \text{with} \quad \hat{U}_t=e^{- \f{i}{\hbar} \hat{H} t}.$$
As $S$ is not a closed system, its state is 
described by the reduced density matrix:
$\varrho_s(t) = \Tr_e  \varrho(t),$ where $\Tr_e$ denotes the partial trace over the degrees of freedom of the environment.
 Indeed  $\varrho_s$ gives the correct statistics of local measurement outcomes on $S$ alone.
The initial total density matrix $\varrho(0)$  can be decomposed as a linear combination of
$|\phi_m\ran \lan \phi_p |$,  implying that $\varrho(t)$ is also a linear combination of 
the density matrices $\hat{U}_t |\phi_m \rangle \langle \phi_p | \hat{U}^\dagger_t $. To calculate the partial trace over the
 environment we can for instance consider the matrix elements of 
 $\hat{U}_t |\phi_m \rangle \langle \phi_p | \hat{U}^\dagger_t $ in the bare eigenbasis $\{ | \phi_1 \ran,.., | \phi_n \ran \}$. By
expanding 
the evolution operator $\hat{U}_t$ over the dressed eigenbasis $\{|\psi_1\ran,...,|\psi_n\ran\}$:
$\hat{U}_t=\sum_{i} e^{-\frac{i}{\hbar } \lambda_it} | \psi_i \rangle \langle \psi_i|,$
these matrix elements $\lan \phi_n | \hat{U}_t    |\phi_m \rangle  \langle \phi_p |  \hat{U}_t^\dagger |\phi_q \ran
$ can be re-written as:
\begin{eqnarray}
\label{FourierOverlaps}
  \nonumber \\  \sum_{i,j}
 e^{-\f{i}{\hbar}(\lambda_i-\lambda_j) t}
  \lan \phi_n  | \psi_i \ran \lan \psi_i | \phi_m\ran 
 \lan \phi_p | \psi_j  \ran  \lan  \psi_j | \phi_q\ran
\end{eqnarray}
and are thus related to the $2-$dimensional Fourier transform of products of four ``overlaps'' $\lan \phi_n | \psi_i \ran$ between the eigenvectors of the bare and dressed Hamiltonians.
In this article, we focus on the long time limit which is the constant given by the case $i=j$ in the summation in Eq.~\eqref{FourierOverlaps}:
\begin{equation}
\label{LongTimeLimit}
\small
\lan \phi_n | \hat{U}_t    |\phi_m \rangle  \langle \phi_p |  \hat{U}_t^\dagger |\phi_q \ran
\xrightarrow[t \to \infty]{}  \sum_{i}   \lan \phi_n  | \psi_i \ran \lan \psi_i | \phi_m\ran 
 \lan \phi_p | \psi_i  \ran  \lan  \psi_i | \phi_q\ran.
\end{equation}
We will focus on the transient regime (given by the summation over the indexes $i \neq j$ in Eq.~\eqref{FourierOverlaps}) and show that it is damped (under the same hypothesis assumed hereafter) in a further publication.
 To calculate the expression in Eq.~\eqref{LongTimeLimit}, one needs an analytical formula for the overlap $\lan \psi_i | \phi_n \ran$ which is typically a quantity accessible in a perturbative framework (i.e. when 
the typical coupling strength is much smaller than the minimum level spacing). In this paper, we use an original method for calculating these quantities which extends in a \textit{statistical} sense the perturbative framework   to \textit{any} interaction strength.

The core idea of the method relies on the hypothesis assumed for the interaction Hamiltonian:
we introduce \textit{deliberately} some randomness in $\hat{W}$ which allows to perform calculations with full generality, knowing that this randomness actually \textit{will not matter} in the large dimensionality limit ($\dim \mathcal{H} Ê\to \infty$) and we will explain why in the sequel.
We proceed as follows. First, we consider the set of Hamiltonians $\hat{W}$ 
whose density of states $\tilde{\rho}_W$ is a given function of energy: $H_{\tilde{\rho}_w}$ which establishes a macroscopic
constraint on the interaction.
We will assume $\tilde{\rho}_W$ is the only a priori knowledge on the interaction we have, 
we know nothing about the microscopic structure of its eigenspaces and this fact provides the origin of randomness
 (see Methods section).
Then we note that each matrix element of 
$\Tr_e \varrho(t)$ is a function of
each matrix element $W_{n,m}=\lan \phi_n |\hat{W} | \phi_m \ran$.
We have a function defined on a high dimensional input space, the set $H_{\rho_W}$, and with complex values. 
In addition, the partial trace is balancing the dependence of 
$\Tr_e \varrho(t)$ on all the $W_{n,m}$.  As a consequence, the reduced density matrix will exhibit a 
phenomenon known as the ``concentration of measure"\cite{MR1387624,MR856576}, 
 which can be described informally as follows: a numerical function that depends in a regular  way on many
   random independent variables, but not too much on any of them, is \emph{essentially constant} and equal to its mean value almost
    everywhere (see the Methods section for a full justification). The mesoscopic fluctuations of this function are squeezed down as the dimension of its input space goes to infinity.
We argue that this phenomenon is at the core of the universality of the thermalisation process: it is responsible for a self-averaging property of the reduced density matrix of $S$, considered as a function of the 
interaction Hamiltonian, meaning that the reduced density matrix of $S$ has a ``typical'' behavior for almost all
interaction Hamiltonians (with given spectrum).
%
It is interesting to compare  this phenomenon, which is a sort of extension of Central Limit Theorems to ``reasonable'' functions, 
to the Monte Carlo method for estimating the integral of a function over some subspace.
A good estimate of this integral is a discrete average of the function sampled randomly over the subspace.
The measure concentration provides a path along the opposite direction: one is interested in the value of a function at 
a \textit{single} point of the subspace and in \textit{most cases}, a very good estimate of this value is the average of the function 
over the subspace. 
As a consequence,  we can compute an analytical value of each terms involving 
$\hat{W}$ simply by averaging:
$\varrho_s(t)= \Tr_e( \varrho(t) ) \approx \mathbb{E} [\Tr_e (\varrho(t))] = \Tr_e \left( \E[\varrho(t)]\right)$, where
$\mathbb{E}$ is the average over the set of interaction Hamiltonians $H_{\rho_W}$ (see the Methods section for a full justification of this averaging procedure).
We are led to consider the average of Eq.~\eqref{LongTimeLimit} which involves the fourth order moments of the overlaps:
$\E[ \lan \phi_n  | \psi_i \ran \lan \psi_i | \phi_m\ran 
 \lan \phi_p | \psi_i \ran \lan \psi_i  |\phi_q\ran]$.

To calculate these moments, 
 we notice the connection between the overlaps and the matrix elements of the resolvent operator $G_H(z)=(H-z \one)^{-1}$ in the eigenbasis 
 of $H_0$: $G_{n,m}(z)=\lan \phi_n | G(z) |\phi_m \ran$, in other words the Green functions.
 Since $G_{n,m}(z)=\sum_i \lan \phi_n | \psi_i \ran \lan \psi_i | \phi_m \ran /(\lambda_i -z)$,
 the product $\lan \phi_n | \psi_i \ran \lan \psi_i | \phi_m \ran$ is the residue of $G_{n,m}(z)$ at the pole $\lambda_i$ (up to
  a factor $2i\pi$). Using an integral representation of the sum in Eq.\eqref{FourierOverlaps} involving these Green functions,
  we can demonstrate that the only non zero cases are when $(n=q \text{ and } m=p)$, which governs the 
  dynamics of the diagonal terms of $\varrho(t)$, or
  when ($n=m$ and $p=q$), which governs the dynamics of the extra diagonal terms of $\varrho(t)$ (see Supp. Info.).
  We will focus in the following on the first case and explain in the Supp. Info. why the stationary regime of the extra diagonal terms can be disregarded (meaning that the quantum coherence is lost in the long time limit as can be expected).
We have now to consider the average of the \textit{diagonal} sum of product of residues
 $\sum_i \E[ |\lan \phi_n | \psi_i \ran |^2  |\lan \phi_m | \psi_i\ran |^2 ]$
whose leading order is provided by the sum over the product of second order moments 
$ \E[ |\lan \phi_n | \psi_i \ran |^2 ] \E[ |\lan \phi_m | \psi_i \ran |^2 ]$. 
Using ``free'' probability tools \cite{VoiculescuFreeConvol,BianeFreeProba}, and
the \textit{subordination} property of the average resolvent $\E[G_H(z)]$ established by Kargin in \cite{KarginSubordination},
  these second order moments can be calculated easily for an interaction with a given density of states and
 isotropically distributed  eigenvectors (see Methods section).
In the limit of large dimension ($\dim \mathcal{H} \to \infty$) and provided the coupling strength $\sigma_w=\sqrt{\tr(W^2)}$ 
(with "$\tr$" the normalized trace)
is much larger than the typical level spacing $1/\tilde{\rho}(\epsilon_n)$ around the energy $\epsilon_n$ (where $\tilde{\rho}$ is the density of states of $\hat{H}$), one has
\begin{eqnarray}
 \E[ |\lan \phi_n | \psi_i \ran |^2 ]
 \approx  \f{\sigma_w^2}{ \dim \mathcal{H}}  \f{1}{(\epsilon_n-\bar{\lambda}_i)^2+ \pi^2 \sigma_w^4 \rho(\epsilon_n)^2 }.
\end{eqnarray}
Note that this last result is identical to the one obtained in~\cite{allez_eigenvectors_2013} for the Gaussian orthogonal matrix ensemble.

These second order moment $\E[ |\lan \phi_m | \psi_i \ran |^2 ]$ defines a local density of states of $|\phi_m\ran$ in the eigenbasis $\{|\psi_1 \ran, ..., |\psi_N\ran\}$: it quantifies how much a bare eigenvector 
is \textit{delocalized} or hybridized in the dressed eigenbasis.
A delocalization ``length'' can then be defined
using the Full Width at Half Maximum (FWHM) of this local density of states and is given by $ \pi \sigma_w^2 \tilde{\rho}(\epsilon_n)/\dim \mathcal{H}$.
Then assuming that the scale of variation of the density of states of
 $\hat{H}$ is much larger than this width, and using  a continuous approximation
($\sum_i \approx \int \tilde{\rho}(\lambda) d \lambda $) we  get for Eq.~\eqref{LongTimeLimit}
the following expression:
\begin{eqnarray}
\sum_i \E[ |\lan \phi_n | \psi_i \ran |^2  |\lan \phi_m | \psi_i\ran |^2 ]
=\bar{p}_{m \to n}  \qquad \qquad  \nonumber \\
\approx \f{1 }{\dim \mathcal{H} } 
  \f{2 \sigma_w^2 }{(\epsilon_n-\epsilon_m)^2+4 \pi^2\sigma_w^4 \rho(\epsilon_n)^2} 
  \label{SumI2} 
\end{eqnarray}
where $\bar{p}_{m \to n}$ defines an \textit{average} transition probability from an initial state 
$|\phi_m \ran \lan \phi_m |$ at $t=0$ to a final state $|\phi_n \ran \lan \phi_n |$ (at $t \to \infty$) which has 
universal properties: it only depends on 
 the density of states of the total Hamiltonian $\hat{H}$ and the strength of the interaction.
  This curve has a remarkable Lorentzian shape 
which is in sharp contrast with the micro canonical hypothesis of equiprobability distribution of the accessible states.  On Fig.\ref{TProb}
 we consider numerical simulations and find a very satisfactory agreement 
with this formula. 



\begin{figure}
\includegraphics[width=8cm]{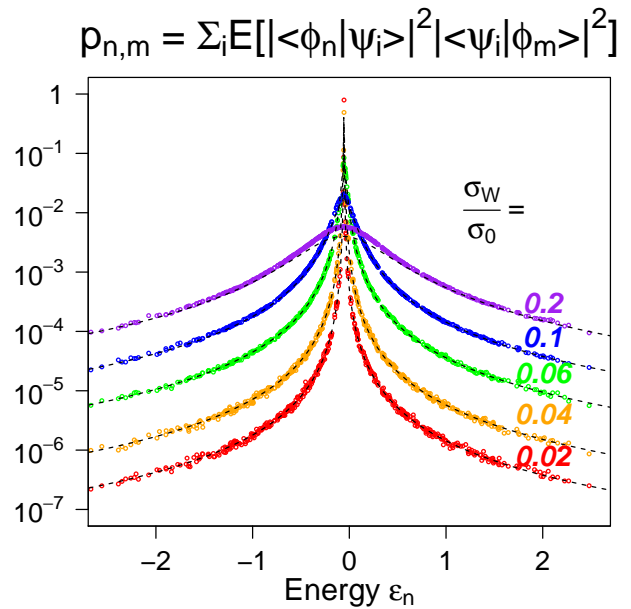}
\caption{ \textbf{Average transition probability $\bar{p}_{m \to n}$ from state $|\phi_m \ran \lan \phi_m |$ at $t=0$ to state $|\phi_n \ran \lan \phi_n |$ at $t\to \infty$.} Comparison between numerical simulations and theoretical predictions. The total Hilbert space size has dimension $\dim \mathcal{H}=500$, and
the density of states of the ``bare'' Hamiltonian $\hat{H}_0$ is assumed to be gaussian distributed with a standard deviation $\sigma_0=1$ and zero mean. The interaction Hamiltonian $\hat{W}$ is an hermitian matrix such that 
$\{\Re(W_{i,j})\}_{i>j}$ and $\{\Im(W_{i,j}) \}_{i>j}$ are independent identically distributed gaussians with standard deviation $\sigma_w/\sqrt{2\dim \mathcal{H}}$ (which provides $\Tr(W^2)/\dim \mathcal{H} = \sigma_w^2$) and zero means.
We plot $\bar{p}_{m \to n}= \sum_i \E[|\lan \phi_n | \psi_i  \ran|^2 | \lan \phi_m | \psi_i \ran|^2]$ for a fixed $m=250$ (middle of the spectrum, $\epsilon_m\approx 0$), as a function of $\epsilon_n$  and for different $S-E$ interaction strengths. This transition matrix $\bar{p}_{m,n}$ is related to how much the eigenvector $|\phi_m \ran$ of $\hat{H}_0$ is \textit{delocalized} in the eigenbasis  of the dressed Hamiltonian because  of the interaction term $\hat{W}$.
The theoretical prediction is provided by Eq.~\eqref{SumI2} and plotted in dashed line. 
When tracing out the degrees of freedom of the environment, this transition probability will ``sample'' the density of states of the environment at the energy $\epsilon_e$ such that $ \epsilon_e+\epsilon_s=\epsilon_m=\epsilon_{e'}+\epsilon_{s'}$.
}
\label{TProb}
\end{figure}


Finally, to perform the partial trace operation and get the reduced density 
matrix of $S$, we recall the final state $|\phi_n \rangle= |\epsilon_{s'} \rangle |\epsilon_{e'} \rangle$
and  the initial state $| \phi_m \rangle = | \epsilon_s \rangle |\epsilon_e \rangle$ 
and sum Eq.~\eqref{SumI2} over $\epsilon_{e'}$ using a continuous approximation: 
$\Tr_e =\sum_{\epsilon_{e'}} \approx \int d\epsilon_{e'}  \; \tilde{\rho}_e(\epsilon_{e'})$.  
Assuming that the scale of variation of the density of states of the environment $\tilde{\rho}_e(\epsilon)$ is much larger than 
the width of $\bar{p}_{m \to n}$
 then the Lorentzian
from Eq.~\eqref{SumI2} is "sampling" $\rho_e(\epsilon)$ at $\epsilon_{e'}=\epsilon_s+\epsilon_e-\epsilon_{s'}$,  which 
provides the main result of this paper: for an initial state 
$\varrho(0)=|\epsilon_s\ran \lan \epsilon_s | \otimes |\epsilon_e \ran \lan \epsilon_e |$,
the long time stationary state has the following distribution
\begin{equation}
\label{ProbaS2}
p_{\epsilon_{s'}}= \lim_{t \to \infty} \langle \epsilon_{s'} | \varrho_s(t) | \epsilon_{s'} \rangle \approx 
\frac{\tilde{\rho}_e (\epsilon_s+\epsilon_e-\epsilon_{s'})}{\tilde{\rho}(\epsilon_s+\epsilon_e)}.
\end{equation}
This ratio provides a very simple interpretation of the thermodynamical equilibrium described on Fig.~(\ref{Fig1}).
\begin{figure}[!t]
  \includegraphics[width=8cm]{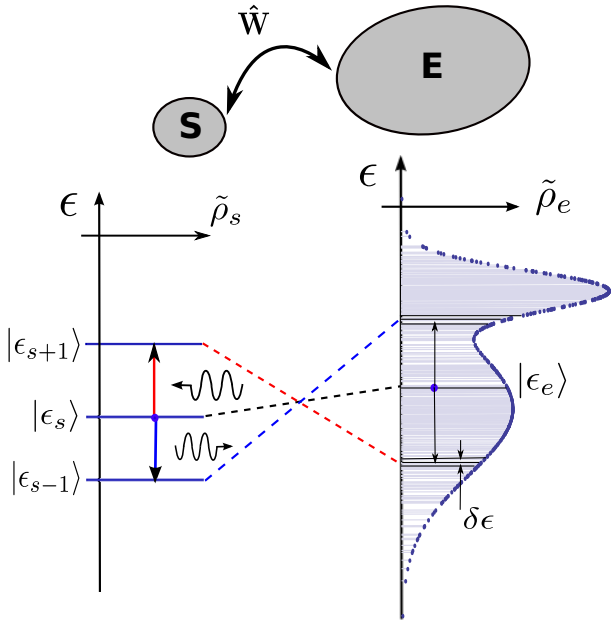}
  \\
  \caption{ \textbf{Summary of the model and its main result on thermalisation}. 
We consider a system $S$  interacting with an environment $E$, both
characterized by the density of states of their respective Hamiltonians ($\tilde{\rho}_s$ for $\hat{H}_s$ and 
$\tilde{\rho}_e$ for $\hat{H}_e$). In the limit of large dimensionality, our model 
 justifies the following interpretation of thermodynamical equilibrium. 
Considering for instance an initial state $ |\Psi_0\ran=|\epsilon_s\ran \otimes  |\epsilon_e\ran$, which is an eigenstate of the "bare" Hamiltonian $\hat{H}_0=\hat{H}_s+\hat{H}_e$, let us ask a first question: 
on average, how many other eigenstates of $\hat{H}_0$ are coupled to $|\Psi_0\ran$ because of the interaction $\hat{W}$? 
The answer is $ \tilde{ \rho}(\epsilon_e+\epsilon_s) \delta \epsilon$ where $\tilde{\rho}$ is the density of states of the \textit{dressed} Hamiltonian $\hat{H}=\hat{H}_0+\hat{W}$ and $\delta \epsilon$ is related to the 
typical width of the \textit{local} density of states of $|\Psi_0\ran$ in the
 eigenbasis $|\psi_i\ran$ of $\hat{H}$: $\E[|\lan \psi_i |\Psi_0\ran|^2]$ (where $\E$ stands for an expectation over the ensemble of interaction Hamiltonians considered).
 This local density of states describes quantitatively how an eigenvector 
 of $\hat{H}_0$ is delocalized in the eigenbasis of $\hat{H}$.
Then let us ask a second question: over this ensemble of states coupled to $|\Psi_0\ran$, \textit{some} of them are such that $S$ is in the state $|\epsilon_{s'}\ran$, 
 how many such states are there? We show that the answer  is $ \tilde{ \rho}_e(\epsilon_b+\epsilon_s-\epsilon_{s'}) \delta \epsilon$.
Finally, in the $t \to \infty$ limit, we show that the probability for $S$ to be in a state $|\epsilon_{s'}\ran$
is the ratio of the later number by the first one:
$p_{\epsilon_{s'}}= \tilde{\rho}_e(\epsilon_e+\epsilon_s-\epsilon_{s'})/\tilde{\rho}(\epsilon_e+\epsilon_s)$, 
which is precisely the result expected from an a priori equal probability of states hypothesis.
Under typical conditions verified by a macroscopic environment, this ratio simplifies to the 
Boltzmann distribution $e^{-\epsilon_{s'}/k_b T}/Z$.
}
  \label{Fig1}
\end{figure}
On the numerator the density of states of the environment is assumed to be known,
 and on the denominator the density of states $\tilde{\rho}(\epsilon)$ of the total Hamiltonian $\hat{H}$ 
 depends in a non trivial way on the density of states of the bare Hamiltonian
  $\tilde{\rho}_0(\epsilon) = [ \tilde{\rho}_s * \tilde{\rho}_e ](\epsilon)=
   \sum_{\epsilon_s }\tilde{\rho}_e(\epsilon-\epsilon_s)$,
 and the density of states of W: $\tilde{\rho}_w$, it is their \textit{free convolution}\cite{VoiculescuFreeConvol,BianeFreeProba}
    $\tilde{\rho}(\epsilon)=\tilde{\rho}_w(\epsilon) \boxplus \tilde{\rho}_0(\epsilon) $ (see Methods section)
which simplifies to $\rho_0$ in the case of weak 
 coupling ($\sigma_w \ll \sigma_0$: the interaction energy is disregarded),
 meaning that the dependence of 
   Eq.~\eqref{ProbaS2}  on the interaction disapear as can be expected.
 Finally, defining the canonical temperature by $\beta=\frac{1}{k_B T}= \frac{d \ln(\rho_e(\epsilon))}{d \epsilon} $,
 and assuming that $\rho_e$ scales exponentially with energy (see \cite{ericson_statistical_1960} for the case of many body interacting quantum systems where this fact relies on the combinatorics of single particle excitations)
  the ratio of densities of states simplifies to the well known canonical distribution:
$$\f{\tilde{\rho}_e(\epsilon_e+\epsilon_s-\epsilon_{s'})}{\tilde{\rho}(\epsilon_e+\epsilon_s)}
\approx \f{\tilde{\rho}_e(\epsilon_e+\epsilon_s-\epsilon_{s'})}{\sum_{\epsilon_{s''} }\tilde{\rho}_e(\epsilon_e+\epsilon_s-
\epsilon_{s''})} \approx \f{e^{-\beta \epsilon_{s'}}}{Z_\beta} $$
with the partition function 
$ Z_\beta=\sum_{\epsilon_s} e^{-\beta \epsilon_s} . $
 All the above results are valid for any initial state of the form
 $|\epsilon_s \ran \lan \epsilon_s | \otimes |\epsilon_e \ran \lan \epsilon_e| $ and can be extended by linearity
  to any initial state, \textit{pure or not}: the initial  extra diagonal terms
(like for instance $|\phi_m \ran \lan \phi_p |$) will not contribute to the final state of $S$, only the initial
 terms $|\phi_m \ran \lan \phi_m|$ contribute. The state of S at long times  will be the weighted average of
  Eq.~\eqref{ProbaS2} by the initial energy distribution of the composite system. 
This leads us to propose to redefine the canonical partition function using the statistics of the overlap coefficients:
 in the long time stationary regime, the probability for the system $S$ to be in a state of energy $\epsilon_s$ is
 $$ p_{\epsilon_s} = \sum_{\epsilon_e,i} \E[ \; \lan \psi_i | \varrho(0) | \psi_i \ran \; 
  |\lan \epsilon_e | \lan \epsilon_s | \psi_i \ran|^2   ] . $$
 This formula only requires the phenomenon of measure concentration to take
 place and does not require the existence of a temperature.


To summarize, 
 a large composite system initially prepared in an arbitrary state and evolving 
unitarily can \textit{locally} converge to a stationary state for a large class of  interaction Hamiltonians.  
The equilibrium properties of a small subsystem
  are fundamentally encoded in the geometrical structure of the eigenspaces of the
  bare and dressed Hamiltonians:
an alternative and more general partition function can be defined 
from the statistics of the overlaps between their eigenvectors. This generalized partition function coincides with the standard canonical partition function in macroscopic conditions.
 The universality of the thermalisation process can be explained by the phenomenon of measure concentration. Reduced density matrices of small subsystems have the property of being \textit{self-averaging},
 which provides a rigorous justification for a new kind of ergodicity: time averages or ensemble averages 
over microscopic states are not required, and can be replaced by ensemble averages over interaction Hamiltonians.

We would like to thank D. Esteve and H. Grabert for their critical reading of the manuscript, their support and the numerous discussions, J.M. Luck for helpful discussions and C. Leruste for his help with the english.




\bibliography{Thermalisation}

\newpage



 \section{Methods}
 \subsection{Random Matrix Theory}
 \la{RandomMat}
 The theory of random matrices was initiated in the 1930's by the statistician Wishart
 and further developed  by several physicists and mathematicians, including Wigner, Dyson~\cite{dyson_statistical_1962} and
  Mehta~\cite{Mehta}  who where, among other subjects,  studying the 
high energy spectra of medium and large weight nuclei (see for instance the review in~\cite{Guhr1998189}). 
Wigner noticed that one cannot infer the Hamiltonian of such a complex $N$-body system from the experimental data. However, since the cross sections could be measured with sufficient energy accuracy
and over a sufficient wide energy range, some statistical study was possible. This led him to make the following drastic change of point of view. He assumes some statistical hypotheses on the Hamiltonian, which are compatible with the general symmetry properties of the system associated to the integrals of motion (energy, spin, charge).
The candidate Hamiltonian for a nucleus can then be written as a block-diagonal matrix where each block corresponds to given values of the conserved quantities (also usually called ``good" quantum numbers).
Then he assumes that the entries of each block of the Hamiltonian are independent identically distributed random variables with variance and mean depending on the conserved quantities associated 
with the block considered.
This allows him to compute averages 
, get \textit{typical} properties for the set of Hamiltonians considered and in particular the average behaviour of energy levels which is of prime importance for nuclear reactions and can be compared to experimental data.
It is interesting to compare Wigner's method with the assumptions made in statistical physics: 
in this framework, the observer renounces the full knowledge of the microscopic state, assuming  the only knowledge that a global constraint is realized (e.g. the energy, particle number, etc... are set to some value), and then, taking a bayesian point of view, assumes that all micro-states compatible with the constraint are equiprobable.
Some thermodynamical quantities can then be calculated and the properties of the canonical state can be established, by considering an ensemble of identically prepared systems all governed by the same Hamiltonian but differing in initial conditions (Gibbs ensemble) and averaging over the ensemble. Then by assuming an ergodic hypothesis, these theoretical predictions can be compared with experimental values which are time averaged quantities of a \textit{single} system.
 As noticed by Metha in \cite{Mehta}, the approach of Wigner is even more radical: there is
a subjective lack of knowledge not only on the state of the system, but also on the \textit{nature} of the system itself. Wigner considers ensembles
of systems driven by different microscopic (at the level of nucleon-nucleon interaction) Hamiltonians which have the same macroscopic (at the level of the whole nucleus) properties.
He thus focuses on typical properties coming from the underlying symmetries associated with the
conserved quantities of the evolution. Dyson justifies this point of view as follows~\cite{Dyson1962a}: ``We picture a complex nucleus as a black box in which a large number of particles are interacting according to unknown laws. As in orthodox statistical mechanics we shall consider an ensemble of Hamiltonians, each of which could describe a different nucleus. There is a strong logical expectation, though no rigorous proof, that an ensemble average will correctly describe the behaviour of one particular system which is under observation". 
The point of view we adopt in this paper is analogous to Wigner and Dyson's, but in our case we can justify the ``ensemble averaging" over the interaction Hamiltonian set by the phenomenon of measure
concentration. This phenomenon provides a rigorous explanation for this new kind of ergodicity.

\subsection{Universality and the phenomenon of measure concentration}
\la{method:conc_of_measure}
The aim of this section is to describe a simple but rather non trivial behavior exhibited by ``reasonable'' real valued functions defined on large dimension spaces: the phenomenon of measure concentration, envisioned in the early work of P. L\'{e}vy but formally 
established by V. Milman and M. Gromov in the $1970$'s and the $1980$'s (see the celebrated paper by Talagrand~\cite{MR1387624}).
The reduced density matrix of the subsystem $S$ at any time $\tau$, $\varrho_s(\tau)$, considered as a function of the interaction $W$, 
exhibits such a behavior.
 Indeed we will see that $\varrho_s(t)$ depends only on some macroscopic quantities summarizing the properties of $W$
and does not depend on the microscopic structure of the eigenspaces of $W$. This phenomenon has two main consequences: it explains the universality of the thermalisation process and it allows to perform analytical calculations with full generality by averaging over
randomness.

\subsubsection{Central Limit Theorem and the Poincar\'{e} inequality}
Every physicist is familiar with the phenomenon of measure concentration in the simplest particular cases described by Central Limit Theorems. Considering for instance an
 experiment whose measurement output is the subject of a random error, then averaging the outcomes of
 $N$ measurements taken in stationary conditions will increase the signal 
 over noise ratio typically by a factor $\sqrt{N}$. This translates mathematically by: having a family of independent identically distributed random variables 
 $\{X_1,..., X_N \}$ with mean $\lan X\ran $ and finite variance $\sigma_X^2$, then the 
 empirical average $f(X_1,...,X_N)=1/N \sum_k X_k$
is a random variable of mean $\lan X \ran$ and standard deviation $\sigma_X/\sqrt{N}$.  
Interestingly, one can notice that if the $X_k$ variables were not independent but perfectly correlated 
(if $X_1=X_2=...=X_N$ for instance) then the standard deviation of $f$ would be $\sigma_f=\sigma_X$.
We see that ``switching off'' the correlations between the $X_k$'s has the important effect of decreasing $\sigma_f$ by a factor $1/\sqrt{N}$.
 
 Let us now consider a slightly more complicated function: the weighted average 
$h(X)= \lan \alpha | X \ran = \sum_k \alpha_k X_k $ where $\alpha=\{\alpha_1,...,\alpha_N\}$  is a deterministic family with $\alpha_i \geq 0$ .
We have $\lan h(X) \ran = ||\alpha ||_{l_1} \lan X \ran$ with the $l_1$ norm $||\alpha||_{l_1}=\sum_k |\alpha_k| $ 
and $\sigma_h^2 = ||\alpha||^2_{l_2} \sigma_X^2$, where
 $||\alpha||^2_{l_2}= \sum_k \alpha_k^2 = || \nabla h||^2_{l_2}$.
The squeezing of the relative fluctuations, if any, goes like  $||\alpha||^2_{l_2} / ||\alpha||_{l_1} $. Clearly here, a condition on the \textit{delocalization} of $\alpha$ is
 necessary for the concentration to take place. Indeed if $\alpha$ is \textit{localized} (i.e. for instance if $\alpha_1=1$ and  $\alpha_k=0$  $\forall k >1$) then
  $||\alpha||_{l_1} \approx ||\alpha||_{l_2}$ and there is obviously no concentration.
On the other side, if $\alpha$ is fully delocalized (i.e. $\alpha_k=1/N$ $\forall k$) then
$h=f$ is concentrated and its relative fluctuations are squeezed like $1/\sqrt{N}$. This delocalization on $\alpha$ sets a
 natural condition for concentration: $h$ should not depend too much on any specific 
$X_k$  in particular in the family $X$, in order to get a sum of incoherent fluctuations and the associated diffusive behavior.
We will see in the following that the mean of the gradient square of $f$ provides a natural way to check this delocalisation
condition.
 
 The next step is obviously to consider more complicated functions than a linear map like $h$. 
 The natural questions are then: what can we say about the statistics of the output of the function $f(X)$ knowing 
the statistics of its input $X$? Is there any king of generalized version of Central Limit Theorems holding for non 
"pathological" functions and probability measures? Surprisingly, there are simple answers to these questions relying on very
 general assumptions on the function $f$ and the probability distribution $P$ of the input. 
The first of them being  the Poincar\'{e} inequality: a probability measure $P$ defined on a domain $\Omega \subset \mathbb{R}^N$
 is said to verify a Poincar\'e inequality if there 
exists a constant $C>0$ such that for all  function $f: \Omega \subset \mathbb{R}^N \rightarrow \mathbb{R}$ continuously differentiable,
 one has
\begin{equation}
\label{PoincareIneq} \sigma_f^2 = \E[ | f(X)  - \E[f] |^2 ] \leqslant  \frac{1}{C} \E[ || \vec{\nabla} f ||_{l_2}^2] 
\end{equation}
where
$$ ||\vec{\nabla} f ||_{l_2}^2 =   \sum_k  \left( \frac{\partial f}{ \partial X_k} \right)^2 
$$
 and where the average $\E$  is relative to $P$.
In other words, the variance of $f$ is \textit{controlled} by the typical gradient strength over a constant $C$ (the Poincar\'{e} constant) usually related to the variance of the input.
 The typical interesting cases are when $C$ does not depend on the dimension and
  $|| \nabla f||^2 \approx 1/N$, for instance with the empirical average case considered above:
  with  $f(X)= \frac{1}{N} \sum_i X_i$, the CLT provides $\sigma_f^2 = \sigma^2/N $,
 since $||\nabla f||_{l_2}^2 = 1/N$ and $1/C=\sigma^2$. 
 The other interesting case is when $1/C=\sigma^2 \approx 1/N$ and $||\nabla f||$ is upper bounded by a constant independent of the dimension. An example for this later case is when $\Omega$ is the hypersphere $\mathbb{S}^{N-1}$ of
  $\mathbb{R}^N$ (i.e. $\{ X \in \mathbb{R}^N / ||X||_{l_2}=1  \}$).
  In both cases, the relative  fluctuations of $f$ around its mean value are
  "squeezed" like $1/\sqrt{N}$ which can be very small when considering input spaces for $f$ made of Hamiltonians operating on large dimension Hilbert spaces, like
   the typical ones of environments. 
  The function is said to be \textit{concentrated} around its mean which is thus a very good estimate of $f$ at any point in a subspace of high dimensionnality.

 In this article, we consider the Poincar\'e inequality to be sufficient for our purpose, however it is possible to 
 go beyond with L\'evy's Lemma (see~\cite{MR856576} p. 141) which provides exponential and Gaussian upper bounds on the statistics of the fluctuations of a function away from its mean value. 
To use the Poincar\'e inequality and show that the reduced density matrix is concentrated we need the Poincar\'e constant 
 for the various probability measures on spaces of interaction Hamiltonian we might consider and
 we need the adequate upper bound on the variance of the gradient of $\rho_s(\tau)$
 considered as a function of $W$.

\subsubsection{Poincar\'e constants for various probability measures on spaces of interaction Hamiltonians}
The Poincar\'e constants are well known for the following probability measures:

\begin{itemize}

\item Matrices with independent real centered Gaussian distributed entries (see for instance Section 4.4 in~\cite{MR2760897}).
The Poincar\'e constant of the Gaussian probability measure $P$ of variance $\sigma^2$ on $\mathbb{R}$ is $1/\sigma^2$. 
It has the property of tensorizing: the Poincar\'e constant for the probability measure $P^{\otimes N}$ defined on $\mathbb{R}^N$ 
is the same: $1/\sigma^2$. The complex case is similar.
As the entries of the matrices we consider have a typical variance $\sigma^2/N$ (to ensure $\tr(\hat{W}^2)=\sigma_w^2$), the Poincar\'{e} constant will scale like $N/\sigma_w^2$.

\item Ensemble of matrices $\{U.D.U^\dagger \}$ where $D$ is a fixed diagonal matrix and $U$ is unitary Haar distributed random
matrix. The Poincar\'e constant is actually related to the Ricci curvature of the ensemble considered as a manifold and the variance of the spectrum of $D$ (see 
appendix F. in~\cite{MR2760897} and the results due to Gromov)
$$C \approx \frac{N}{ 2 \sigma_D^2 } $$
where $ \sigma_D^2 =  \Tr(D^2)/N$ is the variance of the spectrum of $D$ (which is assumed to be fixed).

\end{itemize}

To summarize: in all cases, because the variance of the spectrum $\hat{W}$ is set to a fixed value independent of the dimension, the Poincar\'e constant of the probability measure typically scales like $N$. 

\subsubsection{Concentration of the reduced density matrix $\varrho_s(t)$}
The next step  is to provide an upper bound on the norm of the gradient of $\varrho_s$ which is independent of the dimension of the total
Hilbert space. Recalling the definitions
$$
\small  \varrho_s(t) = \Tr_e( | \psi \ran \lan \psi | )=  \Tr_e(U_t \varrho(0) U_t^\dagger) 
 \text{ and }  U_\tau = e^{-i\frac{t}{\hbar} (H_0+W)}
$$ we show in the 
 Supplementary Information that the gradient of the reduced density matrix at a fixed time considered as a function
of the interaction Hamiltonian is upper bounded:
$$|| \nabla_W \varrho_s ||^2 \leqslant 2 \frac{t^2}{\hbar^2} \dim \mathcal{H}_s.$$
Using the Poincar\'e inequality in Eq.~\eqref{PoincareIneq}, we have the upper bound on the variance of 
the fluctuations of $\varrho_s$ away from its typical behavior :
$$\sigma_\varrho^2  = \E[|| \varrho_s -\E[\varrho_s]  ||^2] \leqslant  \frac{ 4 }{\hbar^2}
  \sigma_W^2 t^2\frac{\dim \mathcal{H}_s}{\dim \mathcal{H}}$$
where the norm $||.||$ refers to the trace norm: $||A||^2= \Tr(A.A^\dagger)$.
For fixed $\dim \mathcal{H}_s, \sigma_w$ and $t$, one has $\sigma_{\varrho} \rightarrow 0$ as the total Hilbert space dimension
goes to infinity. The reduced density matrix is concentrated around a typical behavior. 

\subsection{Free probability and the addition of quantum operators}

Free probability theory provides the tools for solving a problem of crucial interest 
in the context of interacting quantum systems:
what is happening to the eigenvalues and the eigenvectors of an Hermitian matrix $\hat{H}_0$, 
when adding an extra matrix $\hat{W}$ which is \textit{not} a perturbation. 
 For some specific $\hat{H}_0$ and $\hat{W}$, it is a very difficult problem to get information about the spectrum of $\hat{H}+\hat{W}$,
 however when the dimension of the matrices becomes large, this problem gets a remarkable \textit{probabilistic} and simple solution: the spectrum of the \textit{dressed} matrix $\hat{H}_0+\hat{W}$ is essentially the same for \textit{almost} all $\hat{H}_0$ and $\hat{W}$ each with a given spectrum.  
  The ``almost" means that a probability measure on the set of Hermitian matrices with given spectrum has to be defined. The simplest and most natural assumption to be made is that the normalized eigenvectors of, let's say $\hat{W}$,  are distributed isotropically in the eigenbasis of $\hat{H}_0$, meaning that the eigenspaces of $\hat{H}_0$ and $\hat{W}$ are in generic position with respect to each other: there is no preferred direction, $\hat{W}$ has the highest symmetry possible.
  Mathematically, this can be said by rewriting the matrix $\hat{W}$ in the eigenbasis of $\hat{H}_0$ as $\hat{W}=\hat{P}.\hat{D}.\hat{P}^\dagger$ where $\hat{D}$ is diagonal matrix containing the spectrum of $\hat{W}$ and $\hat{P}$ is unitary Haar distributed (its columns are isotropically distributed on the unit sphere of $\mathbb{C}^N$).   
In this case, the matrices $\hat{H}_0$ and $\hat{W}$ have the property of being asymptotically \textit{free} (at first order) with respect to each other (see Voiculescu~\cite{VoiculescuFreeConvol}
  for a strict mathematical definition of first order freeness), 
which is the equivalent for non commutative random variables of the
property of independence for classical (commutative) random variables.
 Under these assumptions, free probability at first order provides the remarkable following tools and results:
\begin{itemize} 
  \item Regarding the spectrum of the dressed matrix $\hat{H}_0+\hat{W}$: the probability density of the spectrum of $\hat{H}_0+\hat{W}$ depends only on the probability density of the spectrum of each $\hat{H}_0$ and $\hat{W}$, and is their \textit{free} convolution: $\rho_{\hat{H}+\hat{W}}(\epsilon)=\rho_{\hat{H}_0}(\epsilon) \boxplus \rho_{\hat{W}}(\epsilon)$. 
  By analogy with classical cumulants of commutative random variables, an $n^{th}$ ``free'' cumulant $\hat{\kappa}_n(\hat{H})$ can
  be defined for the non commutative random operator $\hat{H}$.
The free cumulants $\hat{\kappa}_n(\hat{H})$ and their associated generating function $R_H(z)=\sum_{n>0} \hat{\kappa}_n(\hat{H}) z^{n-1}$ (also called the R-transform) have the property of ``linearizing'' the free convolution, just like the Fourier transform and the classical cumulants linearize the classical convolution between the probability distributions of two \textit{commuting} independent random variables (Voiculescu 1986, Speicher 1994). In other words, no microscopic knowledge of the detailed relation between the eigenspaces of $\hat{H}_0$ and $\hat{W}$ is required to compute the spectrum of  $\hat{H}_0+\hat{W}$, which is remarkable.
Let us also mention that the first three free cumulants of $\hat{W}$ coincide with the \textit{classical} cumulants of the probability density of the spectrum of 
$\hat{W}$.
In the case of a centered Gaussian distributed interaction, all free cumulants vanish except the second one which is the variance of the
spectrum of $\hat{W}$: $\hat{\kappa}_2(\hat{W}) = \sigma_w^2 =\tr(\hat{W}^2)$.
  \item Regarding the eigenvectors of $\hat{H}=\hat{H}_0+\hat{W}$. First, we notice that the product of the overlap coefficients
   $\lan \phi_n | \psi_i\ran \lan \psi_i | \phi_m \ran$ is (up to a factor $2i\pi$) the residue  of the Green function or matrix element 
   $G_{n,m}(z)=\lan \phi_n | G_{\hat{H}}(z) |\phi_m \ran$ at the pole
  $\lambda_i$, since  $G_{n,m}(z) = \sum_i \lan \phi_n | \psi_i\ran \lan \psi_i | \phi_m \ran /(\lambda_i -z)$.
From this relation, one can get the second order moments of the overlaps as a function of 
the average resolvent of $\hat{H}$ and the density of states of $\hat{H}$, $\tilde{\rho}$:
\begin{equation}
\label{MomentsResolvent}
 \E[ |\lan \phi_n | \psi_i \ran |^2 ]
 = \lim_{z \rightarrow \bar{\lambda}_i +i 0^+ }- \f{1}{\pi}  \frac{\Im(\E[G_{n,n}(z)]) }{  \tilde{\rho}(\bar{\lambda}_i)},  
 \end{equation}
 where $\bar{\lambda}_i$ is the mean of the $i^{th}$ eigenvalue. Then assuming that $\hat{W}$ and $\hat{H}_0$
 are asymptotically free at first order,
  Kargin established in~\cite{KarginSubordination} a \textit{subordination} relation between the average of the dressed
  resolvent $\E[G_{\hat{H}}(z)]$ and the bare
  resolvent $G_{\hat{H}_0}(z)$:
\begin{equation}
\label{AverageResolvent}
\E[G_{\hat{H}_0+\hat{W}}(z)]=G_{\hat{H}_0}(z+S_{\hat{W}}(z))  
\end{equation}
where $S_W(z)$, called the ``subordinate'' function, is related to the R-transform of $\hat{W}$ and the Stieltjes transform of $\hat{H}$, $m_H(z)=\tr(G_H(z))$:
$$S_W(z) = R_W(m_H(z)).$$
Finally, combining Eq.\eqref{MomentsResolvent} and \eqref{AverageResolvent}, we get
\begin{eqnarray}
\label{Moments2}
 \E[ |\lan \phi_n | \psi_i \ran |^2 ]
\approx  \f{1}{ \tilde{\rho}(\lambda_i)} \f{1}{\pi} \f{s_{\lambda_i}}{(\epsilon_n-\lambda_i
-\tilde{s}_{\lambda_i})^2+s_{\lambda_i}^2 } 
\end{eqnarray}
where $s_\lambda$ (resp. $\tilde{s}_\lambda$) is the imaginary part (resp. the real part) of the subordinate function, which
defines an effective linewidth  (resp. an effective energy shift) of the bare eigenstate $|\phi_n \ran$
due to the introduction of the additive term $\hat{W}$.
Please note that this result is non perturbative.
To avoid unnecessary complications for understanding the process of thermalisation, in the main part of the paper we truncate the free cumulant expansion of $S$ at first order: $S(z)\approx \sigma_w^2 m_H(z)$ and
recover the case of a  gaussian interaction already calculated in \cite{allez_eigenvectors_2013}:
$s_\lambda=  \sigma_w^2  \rho(\lambda)$ and $\tilde{s}_\lambda=  \sigma_w^2 H_\rho(\lambda) $
where $H_\rho(\lambda)$ is the Hilbert transform of $\rho$.
This later energy shift  due to the Hilbert transform of $\rho$ is typically of the order of one linewidth, and
thus negligable on the scale of variation of the density of states.
To summarize, this subordination property of the resolvent provides a very deep link between the statistics of the eigenvectors of $\hat{H}$ (here the second order moments of the overlaps  $\E[|\lan \phi_n | \psi_i \ran|^2]$)
  and the statistics of the spectra of both $W$ (through its free cumulant expansion) and $\hat{H}$ (through its Stieltjes transform). This has strong implications on the delocalization of the bare eigenvectors (the $ \{ |\phi_nÊ\ran \}_n $) in the dressed eigenbasis 
  (the $ \{ |\psi_iÊ\ran \}_i $) which controls the thermalisation process.
  \end{itemize}
%

\pagebreak

\begin{widetext}

\begin{center}
{\bf SUPPLEMENTARY INFORMATION FOR: \\``THERMALISATION OF A QUANTUM SYSTEM FROM FIRST PRINCIPLES''}
\end{center}

\section{Integral representation of the total density matrix}
\label{IntRep}
Considering the resolvent operator associated to the dressed Hamiltonian $\hat{H}$: $G_H(z)=(\hat{H}-z\one)^{-1}$ defined for
$z \in \mathbb{C}\setminus \mathbb{R}$, the quantum evolution operator $U_t=e^{-i\hat{H}t}$, and subsequently the 
total density matrix, can be rewritten as complex integrals:

\begin{equation}
 \label{UEq}
U_t  = \frac{1}{2i\pi} \oint_C dz G_H(z) e^{-izt}   \quad \text{giving} \quad
 \varrho(t)= U_t \varrho(0) U_t^\dagger = -\frac{1}{4\pi^2} \oint_{C} dz \oint_{C'} dz'  G_H(z)
 \varrho(0) G_H(z') e^{- i(z-z')t} 
 \end{equation}
where $C$ and $C'$ are counter-clockwise \textit{macroscopic} contours circling once around all eigenvalues of $\hat{H}$. 
Decomposing the initial density matrix on the bare basis $\{| \phi_n \ran \}_n$ and averaging over the statistics of $W$, we are lead  to consider the
 correlations between matrix elements of the resolvent (i.e. Green functions): $\E[G_{n,m}(z_1)G_{p,q}(z_2)]$ where
  $G_{n,m}(z)=\lan \phi_n | G_H(z) | \phi_mÊ\ran$ and $\E$ is the average over the statistics of $\hat{W}$.



\section{Zero cases for the first and second order statistics of the resolvent entries}

The aim of this section is to prove that the only non zero correlations between matrix elements of the resolvent operator 
$\E[G_{n,m}(z_1)G_{p,q}(z_2)]$ are when ($n=q$ and $m=p$) or ($n=m$ and $p=q$). 
The first case will contribute to both the transient and the stationary regimes of the \textit{diagonal} terms of the total density matrix, whereas  the last case ($n=m$ and $p=q$) provides both regimes of the extra diagonal terms (or coherences) of $\varrho(t)$. 
In the main part of the paper, we focus on the first case. We provide at the end of this section an explanation why the second case can be neglected.

To proceed, we first remind the results from
 Kargin~\cite{KarginSubordination} on the average of the resolvent operator $G_{\hat{H}}(z)=\frac{1}{\hat{H}- z \one}$, defined for 
 $z\in \mathbb{C}\backslash \mathbb{R}$, where $\hat{H}=\hat{H}_0+\hat{P}.\hat{D}.\hat{P}^\dagger$ with $\hat{H}_0$ a $N \times N$
  hermitian deterministic matrix,  $\hat{P}$ unitary Haar distributed and $\hat{D}$ real deterministic diagonal. The main result being that this average
   is actually diagonal in the eigenbasis of $\hat{H}_0$. Then we adapt Kargin's method to get our result on the second order statistics (the correlations between resolvent entries). 
In the following lemma and its proof, we shall write the matrices in an eigenbasis 
$\{|\phi_1 \rangle,\ldots, |\phi_N \rangle \}$ of $\hat{H}_0$, so that we can suppose that 
$\hat{H}_0$ is a diagonal matrix. Besides, by invariance of the Haar measure (i.e. by the fact that for 
any fixed unitary matrix $\hat{O}$, $\hat{P}\hat{O}$ is also distributed according to the Haar measure), one can suppose that $\hat{D}$ is diagonal. 

\begin{lemma}Let $z, z_1,z_2\in \mathbb{C}\backslash \mathbb{R}$. We have \begin{itemize}
\item[(a)] The matrix $\mathbb{E}[G_{\hat{H}}(z)]$  is diagonal.\\
\item[(b)] The matrix    $\mathbb{E}[G_{\hat{H}} (z_1) \hat{M}  G_{\hat{H}}(z_2)] $ is diagonal for any diagonal matrix $\hat{M}$.\\
\item[(c)] For any $i\ne j$,  for $\hat{M}_{i,j}$ the matrix whose sole non-zero entry is the $(i,j)$-th one, the only non-zero entry of the matrix $\mathbb{E}[G_{\hat{H}} (z_1) \hat{M}_{i,j} G_{\hat{H}} (z_2)] $ is the $(i,j)$-th one.
\end{itemize}
\end{lemma} 

\begin{proof}
Let $p_{ij}$ denote the entries of the unitary matrix $P$ and set $q_{ij}:=p_{ji}^*$ (here, ${}^*$ stands for the 
complex conjugate). Note that in each proposition, by analytic continuation, it suffices to focus on large enough $z,z_1,z_2$. 
In this case, using the formula 
$$ \frac{1}{\hat{H}_0+\hat{P}.\hat{D}.\hat{P}^\dagger-z \one}=- \sum_{k\ge 0}\frac{(\hat{H}_0+\hat{P}.\hat{D}.\hat{P}^\dagger)^k}{z^{k+1}}$$(which is true as soon as $|z|> \|\hat{H}_0\|+\|\hat{D}\|$), the expansion of $\hat{H}_0+\hat{P}.\hat{D}.\hat{P}^\dagger$ and the fact that both $\hat{H}_0$ and $\hat{D}$ are diagonal, the propositions of the lemma reduce  to  the following ones: 
  \begin{itemize}
\item[(a')] For any $p\ge 0$, any $A_1, \ldots, A_{2p}$ diagonal matrices, the matrix    $$ \mathbb{E}[A_1\hat{P}A_2\hat{P}^\dagger \cdots A_{2p-1}\hat{P}A_{2p}\hat{P}^\dagger]$$ is diagonal.
\item[(b')] Same as (a').  
\item[(c')] For any $p,q\ge 0$, any $A_1, \ldots, A_{2p}$, $B_1, \ldots, B_{2q}$ diagonal matrices and any     matrix $M_{i,j}$ whose sole non-zero entry if the $(i,j)$-th one (with $i\ne j$),  the only non-zero entry of the matrix
\small
 $$ \mathbb{E}[A_1\hat{P}A_2\hat{P}^\dagger \cdots A_{2p-1}\hat{P}A_{2p}\hat{P}^\dagger M_{ij}B_1\hat{P}B_2\hat{P}^\dagger \cdots B_{2q-1}\hat{P}B_{2q}\hat{P}^\dagger]$$
 \normalsize is the $(i,j)$-th one.
\end{itemize}
Expanding the matrix products, (a')  and (c') then reduce to
  \begin{itemize}
\item[(a'')] For any $k\ge 1$, any $i_1, \ldots, i_{k+1}$, $j_1, \ldots, j_k$, we have   $$i_1\ne i_{k+1}\implies \mathbb{E}[p_{i_1j_1}q_{j_1i_2} p_{i_2j_2}q_{j_2i_3}\cdots\cdots p_{i_kj_k}q_{j_ki_{k+1}}]=0.$$ 
\item[(c'')] 
For any $k,l\ge 1$, any $i_1, \ldots, i_{k+1}$, $j_1, \ldots, j_k$, any $a_1, \ldots, a_{l+1}$, $b_1, \ldots, b_l$, such that  $i_{k+1}\neq a_1$, we have   
\begin{eqnarray} \nonumber \mathbb{E}[p_{i_1j_1}q_{j_1i_2}  \cdots p_{i_kj_k}q_{j_ki_{k+1}}\times p_{a_1b_1}q_{b_1a_2}  \cdots p_{a_l b_l}q_{b_l a_{l+1}}]\ne 0 \\ \implies i_1=i_{k+1}\textrm{ and }a_1=a_{l+1}. \nonumber
\end{eqnarray}
\end{itemize}

It is easy to see that (c'') reduces to (a''). 
So let us prove (a''). By invariance of the Haar measure by left and right multiplication by diagonal matrices with diagonal entries on the unit circle of $\mathbb{C}$, we know that for the above expectation to be non zero, the number of times an index $i$ appears as first coordinate of a $p_{kl}$ term to be equal to the number of times it appears as the second coordinate of a $q_{kl}$ term. This constraint cannot be satisfied if $i_1\ne i_{p+1}$. 
\end{proof}

\subsection{Stationary regime of the extra diagonal terms of $\varrho(t)$: case ($n=m$ and $p=q$).}
In the long time limit, the matrix element $\E[ \lan \phi_n | U_t | \phi_n \ran  \lan \phi_p | U_t^\dagger | \phi_p \ran]$
is equal to the same diagonal sum of residues as considered in the main body of the text: $\sum_i \E[|\lan \phi_n | \psi_i \ran|^2 |\lan \phi_p | \psi_i \ran |^2]$, however here the partial trace operation 
imposes a strong constraint on $| \epsilon_n -\epsilon_p |$ to get a non zero contribution to the extra diagonal matrix element of the reduced density matrix of $S$: $| \epsilon_n -\epsilon_p |$ is always larger that $ D_s $ the minimum level spacing of the system $S$ alone. 
As a consequence, the extra diagonal terms of the reduced density matrix of $S$ in the long time limit (i.e. $\lim_{t \to \infty} 
\lan \epsilon_s | \varrho_s(t) | \epsilon_{s'} \ran$ for $s\neq s'$)  will always be 
smaller than their diagonal counterparts (i.e. $\lim_{t \to \infty} \lan \epsilon_s | \varrho_s(t) | \epsilon_{s} \ran$)
by a factor equal to $\pi^2 \sigma_w^4 \rho^2$ (the FWHM of the $\bar{p}_{m \to n}$ curve ) over the $D_s$.
Under our assumptions, this factor is very small. 


\section{Concentration of the reduced density matrix of the subsystem $S$}

\label{SupInfoConcentration}

In this section, our aim is to provide an upper bound on the norm of the gradient of $\rho_s$ with respect to $W$, which is uniform in the dimension. This upper bond is used in the Methods section to show that the reduced density matrix of $S$ is concentrated around a typical behavior.

Recalling the definitions
$$  | \psi \ran \lan \psi | = U_\tau \rho(0) U_\tau^\dagger
\qquad \text{and} \qquad U_\tau = e^{-i\tau (\hat{H}_0+\hat{W})},$$

we use the formula for the differential of the exponential map, in order to get the differential of $|\psi \ran \lan \psi |$ with respect to the interaction $\hat{W}$:

$$
\textbf{d} |\psi \ran \lan \psi | (\delta W) =  \int_0^1  h \circ g_\alpha (\delta W) d\alpha   
$$
with 
 $$ h(A)= -i \tau  [A , |\psi\ran \lan \psi | ]  \qquad \text{and} \qquad   g_\alpha(B)= U_{\alpha \tau} \; B \; U_{\alpha \tau }^\dagger \label{Comm} 
$$

where $[,]$ denotes the commutator and $\circ$ the composition of functions.
We start by the upper bond
$$ \small 
|| \nabla_W \rho_s ||^2  =\left\Vert  \int_0^1  \Tr_e h \circ g_\alpha(\cdot) d\alpha \right\Vert^2  \leqslant  \int_0^1 || \Tr_e h \circ g_\alpha(\cdot)||^2 d\alpha $$

where the notation $||\cdot ||$ is for the norm defined on the ensemble of linear applications from the space of interaction Hamiltonians to the space of density matrices on $\mathcal{H}_s$:

$$||f||^2= \sum_{i,j} ||f(M_{i,j})||_F^2 \quad \text{where} \quad ||A||_F^2=\Tr(A.A^\dagger),$$
and
 $$  f : W \in H_{\dim \mathcal{H} \times \dim \mathcal{H}}(\mathbb{C}) \rightarrow  \rho_s \in H_{\dim \mathcal{H}_s \times \dim \mathcal{H}_s}(\mathbb{C}),$$
$H_{n,n}(\mathbb{C})$ is the ensemble of hermitian matrices of size $n \times n$
and  $M_{i,j}$ is the matrix with zero everywhere except at the intersection of the $i^{th}$ line and $j^{th}$ column.
Then we have the equality:

$$  \int_0^1 || \Tr_e h \circ g_\alpha(\cdot)||^2 d\alpha =  ||\Tr_e h(\cdot)||^2$$

since $g_\alpha$ is a unitary change of orthonormal basis.  To move on and make the partial trace easy, we write $|\psi \ran $ in the tensor basis $|\psi \ran= \sum_{s,e} \gamma_{s,e} \; |\epsilon_s \ran \otimes |\epsilon_e \ran$ where
$\{ \gamma_{s,e} \}_{s,e}$ is a $\dim \mathcal{H}_s \times \dim \mathcal{H}_e $ matrix.
Then we have for the matrix elements of $\Tr_e h(M_{a,b,c,d})$ where $M_{a,b,c,d}= |s_a \ran \lan  s_b| \otimes  | e_c \ran \lan e_d |$:

$$ \lan \epsilon_s| \Tr_e [M_{a,b,c,d}, | \psi \ran \lan \psi|] |\epsilon_{s'} \ran =  \gamma_{b,d} \; \gamma_{s',c}^* \;  \delta_{s,a}  - \gamma_{s,d} \; \gamma_{a,c}^* \;   \delta_{b,s'} $$

Taking the square modulus and summing over $a,b,c,d$:

\begin{eqnarray}
\small \nonumber
\sum_{a,b,c,d}  | \lan \epsilon_s | \Tr_e ( [ M_{a,b,c,d},|\psi \ran \lan \psi | ] ) | \epsilon_{s'} \ran |^2 =
 \sum_{b\neq s',c,d} | \gamma_{b,d}  |^2 | \gamma_{s',c}|^2 \nonumber \\
  + \sum_{a\neq s,c,d} | \gamma_{s,d}  |^2 |\gamma_{a,c}|^2  \nonumber 
 + \sum_{c,d}  | \gamma_{s',d} \gamma_{s',c} - \gamma_{s,d} \gamma_{s,c} |^2 \nonumber  \\
 = 
   \left(  \sum_{b,c,d} | \gamma_{b,d}  |^2 | \gamma_{s',c}|^2 + \sum_{a,c,d} | \gamma_{s,d}  |^2 |\gamma_{a,c}|^2 \right) \nonumber\\
  - 2  \sum_{c,d}\Re\left(  \gamma_{s',d} \; \gamma_{s',c} \; \gamma_{s,d}^* \; \gamma_{s,c}^* \right)  \nonumber
 \end{eqnarray}

then summing over $s$ and $s'$, we get the square of the norm we are looking for:

\begin{eqnarray}
\nonumber
||\Tr_e h(\cdot) ||^2 & =&  \tau^2 \sum_{s,s'} \sum_{a,b,c,d}  \left| \lan \epsilon_s | \Tr_e ( [ M_{a,b,c,d},|\psi \ran \lan \psi | ] )  | \epsilon_{s'} \ran \right|^2 
   \\ \nonumber
& = & 2 \tau^2 \dim \mathcal{H}_s -   2\tau^2 \sum_{c,d,s,s'}\Re\left(  \gamma_{s',d} \gamma_{s',c} \gamma_{s,d}^* \gamma_{s,c}^* \right)
 \end{eqnarray}
 since the normalization condition: $ \Tr(\gamma.\gamma^\dagger)  = 1$. The second term on the right hand side is:

$$\sum_{c,d,s,s'} \left(  \gamma_{s',d} \gamma_{s',c} \gamma_{s,d}^* \gamma_{s,c}^* \right)= \Tr((\gamma.\gamma^\dagger)^2) \geqslant 0$$

We finally get that 

$$|| \nabla_W \rho_s ||^2 \leqslant || \Tr_e h (\cdot)  ||^2  \leqslant 2 \tau^2 \dim \mathcal{H}_s.$$


\end{widetext}

\end{document}